\documentclass[10pt]{article}%
\usepackage{amssymb}
\usepackage{amsfonts}
\usepackage{amsmath}
\usepackage{graphicx}
\usepackage[colorlinks,linkcolor={blue},  bookmarks={false},
pdfpagelayout={SinglePage}  ]{hyperref}
\usepackage[UKenglish]{isodate}%
\newtheorem{theorem}{Theorem}
\newtheorem{definition}[theorem]{Definition}
\newtheorem{lemma}[theorem]{Lemma}
\newtheorem{proposition}[theorem]{Proposition}
\newtheorem{remark}[theorem]{Remark}
\newenvironment{proof}[1][Proof]{\noindent\textbf{#1.} }{\ $\blacksquare$}
\begin{document}

\title{Construction of Martingale Measure in the Hazard Process Model of Credit Risk}
\author{Marek Capi\'{n}ski\thanks{\texttt{capinski@agh.edu.pl}; Faculty of
Mathematics, AGH--University of Science and Technology, Al.~Mickiewicza~30,
30--059 Krak\'ow, Poland.}\ \ and Tomasz
Zastawniak\thanks{\texttt{tomasz.zastawniak@york.ac.uk}; Department of
Mathematics, University of York, Heslington, York YO10~5DD, United Kingdom.}}
\date{\cleanlookdateon\today}
\maketitle

\begin{abstract}
In credit risk literature, the existence of an equivalent martingale measure
is stipulated as one of the main assumptions in the hazard process model. Here
we show by construction the existence of a
measure that turns the discounted stock and defaultable bond prices into
martingales by identifying a no-arbitrage condition, in as weak a sense as
possible, which facilitates such a construction.

\end{abstract}

\section{Introduction}

No arbitrage is the principal condition in mathematical finance, a basis for
pricing derivative securities. In the literature on the hazard process model
of credit risk (for example, \cite{BieRut2002}, \cite{BieJeaRut2009} and
references therein) the existence of an equivalent martingale measure is
assumed, the lack of arbitrage following as an immediate consequence. Here we
work in the opposite direction.

A construction of a martingale measure in the relatively straightforward case
of the hazard function model of credit risk in the absence of simple arbitrage
was accomplished in~\cite{CapZas2014}. In the much more general setting of the
hazard process model, the construction of a martingale measure from a suitably
weak no-arbitrage condition turns out to be far from trivial, and constitutes
the main result of the present paper. This no-arbitrage condition, referred to
as the no-quasi-simple-arbitrage principle later in the paper, implies that
the pre-default value of the defaultable bond is a strict submartingale with
values between~$0$ and~$1$ under the Black--Scholes measure. It makes it
possible to apply the Doob--Meyer type multiplicative decomposition for
positive submartingales, leading to a new definition of the survival process,
hence of the hazard process, as the unique up to indistinguishability strictly
positive previsible (with respect to the Black--Scholes filtration) process
that features in the multiplicative decomposition. The martingale measure in
the hazard process model is then constructed with the aid of the survival (or
hazard) process by a method resembling the classical construction of Wiener
measure on path space.

\section{Market model\label{Sect:f745hhandta}}

We consider three assets, a non-defaultable bond $B(t,T)=e^{-r\left(
T-t\right)  }$ for $t\in\lbrack0,T]$ growing at a constant rate $r\geq0$, a
stock with prices $S(t)$ for $t\geq0$, and a defaultable bond with prices
$D(t,T)$ for $t\in\lbrack0,T]$, where $T>0$ is the maturity date for both
bonds. The price processes of the risky assets $S(t)$ and $D(t,T)$ are defined
on a probability space $\left(  \Omega,\Sigma,P\right)  $, where $P$ is the
physical probability. Throughout this paper, equalities and inequalities
between random variables on $\left(  \Omega,\Sigma,P\right)  $ as well as
pathwise properties of stochastic processes such as, for example, continuity
of paths will be understood to hold $P$-a.s.

The stock price process $S(t)$ is assumed to follow the Black--Scholes model
with driving Brownian motion~$W(t)$. We write $\left(  \mathcal{F}_{t}\right)
_{t\geq0}$ for the augmented filtration generated by the Brownian
motion.

We also take a random variable $\tau>0$ on $\left(  \Omega,\Sigma,P\right)  $
to play the role of the time of default. Let $\left(  \mathcal{I}_{t}\right)
_{t\geq0}$ be the filtration generated by the default indicator process
$I(t)=\mathbf{1}_{\left\{  \tau\leq t\right\}  }$ and let $\left(
\mathcal{G}_{t}\right)  _{t\geq0}$ be the enlarged filtration,%
\[
\mathcal{G}_{t}:=\sigma(\mathcal{F}_{t}\cup\mathcal{I}_{t})
\]
for each $t\geq0$.

The defaultable bond price process is assumed to be of the form%
\begin{equation}
D(t,T)=c(t)\mathbf{1}_{\left\{  t<\tau\right\}  }, \label{ahewtfd54bd7}%
\end{equation}
where $c(t)$, $t\in\lbrack0,T]$, called the \emph{pre-default value} of
$D(t,T)$, is an $(\mathcal{F}_{t})_{t\in\lbrack0,T]}$-adapted process with
continuous paths
such that $c(t)\in(0,1)$ for $t\in\lbrack0,T)$ and $c(T)=1$. In particular,
the payoff of this bond is $D(T,T)=\mathbf{1}_{\left\{  T<\tau\right\}  }$ at
time~$T$, that is, the defaultable bond has zero recovery.

\begin{remark}
\upshape In Appendix~\ref{Sect:str5eydvshs} we show that expression
(\ref{ahewtfd54bd7}) follows from certain weaker assumptions by means of a
no-arbitrage argument within a family of simple strategies in the $BD$ section
of the market only. This is similar to the hazard function model consisting of
two bonds $B$ and $D$ only, considered in~\cite{CapZas2014}. However, while
$c(t)$ is a deterministic strictly increasing function in the toy model
in~\cite{CapZas2014}, here it is an $(\mathcal{F}_{t})_{t\in\lbrack0,T]}%
$-adapted process, which turns out to be a strict supermartingale under a
suitable no-arbitrage condition as shown in Section~\ref{Sect.oape9md6bf54}.
\end{remark}

Additionally, we assume that
\begin{equation}
P\left(  s<\tau\leq t|\mathcal{F}_{T}\right)  >0 \label{d7465edgd}%
\end{equation}
for any $s,t\geq0$ such that $s<t$. In other words, for every $A\in
\mathcal{F}_{T}$ of positive measure~$P$, the event $A\cap\left\{  s<\tau\leq
t\right\}  $ is also of positive measure~$P$. This condition means that there
are no gaps in the set of values of~$\tau$, i.e.\ default can happen at any
time, no matter what the stock price process is doing.\footnote{The following
slightly weaker condition is in fact sufficien
t: $P\left(  T<\tau
|\mathcal{F}_{T}\right)  >0$ and $P\left(  s<\tau\leq t|\mathcal{F}%
_{T}\right)  >0$ for all $s,t\in\lbrack0,T]$ such that $s<t$.}

Since the stock $S$ follows the Black--Scholes model, there is a unique
probability measure~$Q_{BS}$ equivalent to~$P$ such that the discounted stock
price process $e^{-rt}S(t)$ is an $\left(  \mathcal{F}_{t}\right)  _{t\geq0}%
$-martingale under~$Q_{BS}$. Since all the processes will need to be
considered up to time~$T$ only, it will suffice if $Q_{BS}$ is understood as a
measure defined on the $\sigma$-algebra~$\mathcal{F}_{T}$.

\section{No quasi-simple arbitrage\label{Sect.oape9md6bf54}}

In the $BS$ segment of the market a self-financing strategy with rebalancing
in continuous time can be defined in the usual manner in terms of the
stochastic integral with respect to the Black--Scholes stock price
process~$S(t)$. Such a strategy is said to be \emph{admissible} whenever its
discounted value process is an $(\mathcal{F}_{t})_{t\ge0}$-martingale
under~$Q_{BS}$.

On the other hand, the above properties of the defaultable bond are \emph{a
priori} not enough to consider a stochastic integral with respect to the
process $D(t,T)$. Hence, for the time being at least, we consider a class of
self-financing strategies such that continuous rebalancing is allowed within
the $BS$ segment of the market, while the position in~$D$ can only be
rebalanced at a finite set of times. We will show that lack of arbitrage
opportunities within the class of such strategies is equivalent to the
existence of a martingale measure.

A strategy of this kind can be constructed as follows. Take $0=s_{0}%
<s_{1}<\cdots<s_{N}=T$ to be the defaultable bond rebalancing times for the
strategy, and let $y_{n}$ be $\mathcal{F}_{s_{n-1}}$-measurable random
variables representing the positions in~$D$ within the time intervals
from~$s_{n-1}$ to~$s_{n}$ for $n=1,\ldots,N$. At time~$0$ we start an
admissible self-financing Black--Scholes strategy $x_{1}=(x_{1}^{B},x_{1}%
^{S})$ in the $BS$ segment of the market, and follow this strategy up to
time~$s_{1}$. Then, if $\tau\leq s_{1}$, that is, if default has already
occurred and the defaultable bond $D$ has become worthless, we follow the same
strategy~$x_{1}$ up to time~$T$. But if $s_{1}<\tau$, that is, if no default
has occurred yet, we rebalance the position in~$D$ from $y_{1}$ to~$y_{2}$,
which means a (positive or negative) cash injection into the $BS$ segment of
the market. We add this cash injection to the value of the strategy~$x_{1}$ at
time~$s_{1}$, and start a new self-financing Black--Scholes strategy
$x_{2}=(x_{2}^{B},x_{2}^{S})$ from this new value at time~$s_{1}$. Then, at
time~$s_{2}$ we either continue following the same strategy~$x_{2}$ up to time
$T$ if $\tau\leq s_{2}$, or else we rebalance the position in~$D$ from~$y_{2}$
to~$y_{3}$, adjust the value of the $BS$ segment accordingly, and start a new
self-financing Black--Scholes strategy $x_{3}=(x_{3}^{B},x_{3}^{S})$ from the
adjusted value at time~$s_{2}$. In this manner, we proceed step by step up to
and including time~$s_{N-1}$. This is formalised in the next definition.

\begin{definition}
\label{Def:957dh45d1a}\upshape By a \emph{quasi-simple self-financing
strategy} we understand an $\mathbb{R}^{3}$-valued $(\mathcal{G}_{t}%
)_{t\in[0,T]}$-adapted process $\varphi=\left(  \varphi^{B},\varphi
^{S},\varphi^{D}\right)  $ representing positions in $B,S,D$ such that there
are sequences of times $0=s_{0}<s_{1}<\cdots<s_{N}=T$, $\mathbb{R}^{2}$-valued
$(\mathcal{F}_{t})_{t\in[0,T]}$-adapted processes $x_{1},\ldots,x_{N}$ and
$\mathbb{R}$-valued random variables $y_{1},\ldots,y_{N}$ satisfying the
following conditions:

\begin{enumerate}
\item \label{Def:QSSFS:item1}$x_{n}=\left(  x_{n}^{B},x_{n}^{S}\right)  $ is
an admissible self-financing Black--Scholes strategy in the time interval
$[s_{n-1},T]$ and $y_{n}$ is an $\mathcal{F}_{s_{n-1}}$-measurable random
variable such that%
\begin{equation}
\varphi^{B}(t)=x_{n\wedge\mu}^{B}(t),\quad\varphi^{S}(t)=x_{n\wedge\mu}%
^{S}(t),\quad\varphi^{D}(t)=y_{n\wedge\mu} \label{Eq:nf6sfsfra}%
\end{equation}
for each $n=1,\ldots,N$ and $t\in(s_{n-1},s_{n}]$, where%
\[
\mu:=\max\left\{  m=1,\ldots,N:s_{m-1}<\tau\right\}  ;
\]

\item \label{Def:QSSFS:item4}The \emph{value process}%
\[
V_{\varphi}(t):=\varphi^{B}(t)B(t,T)+\varphi^{S}(t)S(t)+\varphi^{D}(t)D(t,T)
\]
satisfies the following \emph{self-financing condition} for each
$n=0,\ldots,N-1$:%
\[
V_{\varphi}(s_{n})=\lim_{t\searrow s_{n}}V_{\varphi}(t).
\]

\end{enumerate}
\end{definition}

\begin{remark}
\upshape The minimum $n\wedge\mu$ in~(\ref{Eq:nf6sfsfra}) captures the fact
that we switch to a new Black--Scholes strategy~$x_{n}$ and a new defaultable
bond position~$y_{n}$ at time~$s_{n-1}$ only if no default has yet occurred at
that time.
\end{remark}

\begin{definition}
\upshape We say that the \emph{no-quasi-simple-arbitrage} (NQSA)
\emph{principle} holds if there is no quasi-simple self-financing
strategy~$\varphi=\left(  \varphi^{B},\varphi^{S},\varphi^{D}\right)  $ such
that $V_{\varphi}(0)=0$, $V_{\varphi}(T)\geq0$, and $V_{\varphi}(T)>0$ with
positive probability~$P$.
\end{definition}

The following result provides a characterisation of the {NQSA} principle in
terms of the process~$c(t)$.

\begin{theorem}
\label{Thm:d64gsna90ms}Under the assumptions in Section~\ref{Sect:f745hhandta}%
, the following conditions are equivalent:

\begin{enumerate}
\item \label{Thm:NQSA:item1}The {NQSA} principle holds;

\item \label{Thm:NQSA:item2}The process $e^{-rt}c(t)$, $t\in\lbrack0,T]$ is a
strict $\left(  \mathcal{F}_{t}\right)  _{t\in\lbrack0,T]}$-submartingale
under~$Q_{BS}$.
\end{enumerate}
\end{theorem}

\begin{proof}
Since we can work with discounted values, it is enough to consider the case
when $r=0$, that is, $B(t,T)=1$ for all $t\in\lbrack0,T]$.

We begin by showing that \ref{Thm:NQSA:item1}$\,\Rightarrow\,$%
\ref{Thm:NQSA:item2}. Suppose that the {NQSA} principle holds. To verify that
$c(t)$ is a strict $\left(  \mathcal{F}_{t}\right)  _{t\in\lbrack0,T]}%
$-submartingale under~$Q_{BS}$, we take any $t_{1},t_{2}\in\lbrack0,T]$ such
that $t_{1}<t_{2}$ and need to show that $c(t_{1})<\mathbb{E}_{Q_{BS}}\left(
c(t_{2})|\mathcal{F}_{t_{1}}\right)  $. Let
\[
A:=\left\{  c(t_{1})\geq\mathbb{E}_{Q_{BS}}\left(  c(t_{2})|\mathcal{F}%
_{t_{1}}\right)  \right\}  .
\]
Because $c(t_{2})$ is a random variable with values in $(0,1]$, it is square
integrable. Therefore, there exists an admissible self-financing strategy
$x=(x^{B},x^{S})$ in the Black--Scholes model that replicates the contingent
claim~$c(t_{2})$ at time~$t_{2}$, that is,%
\[
x^{B}(t_{2})+x^{S}(t_{2})S(t_{2})=c(t_{2}).
\]
The value of the strategy%
\[
x^{B}(t)+x^{S}(t)S(t)=\mathbb{E}_{Q_{BS}}\left(  c(t_{2})|\mathcal{F}%
_{t}\right)
\]
for any $t\in\lbrack0,t_{2}]$ is an $\left(  \mathcal{F}_{t}\right)
_{t\in\lbrack0,t_{2}]}$-martingale under~$Q_{BS}$ and has continuous paths. We
extend this strategy by putting $x(t):=(c(t_{2}),0)$ for any $t\in(t_{2},T]$.
Using this, we can construct a quasi-simple self-financing strategy as follows:

\begin{itemize}
\item Do nothing until time $t_{1}$.

\item At time $t_{1}$, if the event $A$ has occurred but no default has
happened yet, that is, $t_{1}<\tau$, then sell a single defaultable bond~$D$
for $c(t_{1})$, invest the amount $x^{B}(t_{1})+x^{S}(t_{1})S(t_{1}%
)=\mathbb{E}_{Q_{BS}}\left(  c(t_{2})|\mathcal{F}_{t_{1}}\right)  $ in the
strategy $x=(x^{B},x^{S})$, and put the balance of these transactions into the
non-defaultable bonds~$B$. Then follow the self-financing strategy
$x=(x^{B},x^{S})$ in the $BS$ segment of the market up to time~$t_{2}$.
Otherwise do nothing.

\item At time $t_{2}$ close all positions and invest the balance in the
non-defaultable bons $B$ until time\ $T$.
\end{itemize}

\noindent The precise formulas defining this strategy $\varphi=(\varphi
^{B},\varphi^{S},\varphi^{D})$ are%
\[%
\begin{tabular}
[c]{ll}%
$\varphi^{B}(t):=\varphi^{S}(t):=\varphi^{D}(t):=0$ & $\text{for }t\in
\lbrack0,t_{1}],$\\
& \\
$\varphi^{B}(t):=\left(  x^{B}(t)+c(t_{1})-\mathbb{E}_{Q_{BS}}\left(
c(t_{2})|\mathcal{F}_{t_{1}}\right)  \right)  \mathbf{1}_{A\cap\left\{
t_{1}<\tau\right\}  },\quad\quad$ & \\
$\varphi^{S}(t):=x^{S}(t)\mathbf{1}_{A\cap\left\{  t_{1}<\tau\right\}
},\ \varphi^{D}(t):=-\mathbf{1}_{A\cap\left\{  t_{1}<\tau\right\}  }$ &
$\text{for }t\in(t_{1},t_{2}],$\\
& \\
\multicolumn{2}{l}{$\varphi^{B}(t):=c(t_{2})\mathbf{1}_{A\cap\left\{
t_{1}<\tau\leq t_{2}\right\}  }+\left(  c(t_{1})-\mathbb{E}_{Q_{BS}}\left(
c(t_{2})|\mathcal{F}_{t_{1}}\right)  \right)  \mathbf{1}_{A\cap\left\{
t_{1}<\tau\right\}  },$}\\
$\varphi^{S}(t):=0,\ \varphi^{D}(t):=-\mathbf{1}_{A\cap\left\{  t_{1}<\tau\leq
t_{2}\right\}  }$ & $\text{for }t\in(t_{2},T].$%
\end{tabular}
\]
Consider the case when $0<t_{1}<t_{2}<T$ (the other cases when $t_{1}=0$ or
$t_{2}=T$ are similar and will be omitted for brevity). In
Definition~\ref{Def:957dh45d1a} we take $N:=3$ and $s_{0}:=0$, $s_{1}:=t_{1}$,
$s_{2}:=t_{1}$, $s_{3}:=T$. We also put%
\[%
\begin{tabular}
[c]{ll}%
$x_{1}(t):=(0,0),$ & $y_{1}:=0,$\\
$x_{2}(t):=1_{A}(x^{B}(t)+c(t_{1})-\mathbb{E}_{Q_{BS}}\left(  c(t_{2}%
)|\mathcal{F}_{t_{1}}\right)  ,x^{S}(t)),$ & $y_{2}:=-1_{A},$\\
$x_{3}(t):=1_{A}(c(t_{1})-\mathbb{E}_{Q_{BS}}\left(  c(t_{2})|\mathcal{F}%
_{t_{1}}\right)  ,0),$ & $y_{3}:=0.$%
\end{tabular}
\
\]
This gives the above strategy $\varphi=(\varphi^{B},\varphi^{S},\varphi^{D})$.
Its initial value is $V_{\varphi}(0)=0$ and final value is%
\[
V_{\varphi}(T)=c(t_{2})\mathbf{1}_{A\cap\left\{  t_{1}<\tau\leq t_{2}\right\}
}+\left(  c(t_{1})-\mathbb{E}_{Q_{BS}}\left(  c(t_{2})|\mathcal{F}_{t_{1}%
}\right)  \right)  \mathbf{1}_{A\cap\left\{  t_{1}<\tau\right\}  }.
\]
Since the {NQSA} principle holds and $c(t_{1})\geq\mathbb{E}_{Q_{BS}}\left(
c(t_{2})|\mathcal{F}_{t_{1}}\right)  $ on~$A$, we must have $P\left(
A\cap\left\{  t_{1}<\tau\leq t_{2}\right\}  \right)  =0$ given that
$c(t_{2})>0$. Because $A\in\mathcal{F}_{s}\subset\mathcal{F}_{T}$, it follows
by assumption~(\ref{d7465edgd}) that $P\left(  A\right)  =0$, proving that
$c(t)$ is indeed a strict $(\mathcal{F}_{t})_{t\in\lbrack0,T]}$-submartingale
under~$Q_{BS}$. This completes the proof of the implication
\ref{Thm:NQSA:item1}$\,\Rightarrow\,$\ref{Thm:NQSA:item2}.

To prove the implication \ref{Thm:NQSA:item2}$\,\Rightarrow\,$%
\ref{Thm:NQSA:item1}, we assume that $c(t)$ is a strict $(\mathcal{F}%
_{t})_{t\in\lbrack0,T]}$-submartingale under~$Q_{BS}$, and take any
quasi-simple self-financing strategy~$\varphi=\left(  \varphi^{B},\varphi
^{S},\varphi^{D}\right)  $ with $V_{\varphi}(0)=0$ and $V_{\varphi}(T)\geq0$.
To verify that the {NQSA} principle holds, we need to show that $V_{\varphi
}(T)=0$. Let $s_{n},x_{n},y_{n}$ be the corresponding sequences as in
Definition~\ref{Def:957dh45d1a}.

For each $n=1,\ldots,N$ and $t\in\lbrack s_{n-1},T]$, we put%
\[
U_{n}(t):=x_{n}^{B}(t)+x_{n}^{S}(t)S(t).
\]
Then, for each $n=1,\ldots,N$ and $t\in(s_{n-1},s_{n}]$, we have%
\begin{align}
V_{\varphi}(t)  &  =\varphi^{B}(t)+\varphi^{S}(t)S(t)+\varphi^{D}%
(t)D(t,T)=U_{n\wedge\mu}(t)+y_{n\wedge\mu}c(t)\mathbf{1}_{\left\{
t<\tau\right\}  }\nonumber\\
&  =\sum_{k=1}^{n-1}U_{k}(t)\mathbf{1}_{\left\{  s_{k-1}<\tau\leq
s_{k}\right\}  }+U_{n}(t)\mathbf{1}_{\left\{  s_{n-1}<\tau\right\}  }%
+y_{n}c(t)\mathbf{1}_{\left\{  t<\tau\right\}  }. \label{Eq:nv74oa0m}%
\end{align}

First, we consider the case when $V_{\varphi}(s_{n})\geq0$ for each
$n=0,\ldots,N$, and proceed by induction on~$n$ to show that $V_{\varphi
}(s_{n})=0$ for each $n=0,\ldots,N$. For $n=0$ we have $V_{\varphi}%
(s_{0})=V_{\varphi}(0)=0$. Suppose that $V_{\varphi}(s_{n-1})=0$ for some
$n=1,\ldots,N$. Self-financing at~$s_{n-1}$ means that%
\begin{align*}
0  &  =V_{\varphi}(s_{n-1})=\lim_{t\searrow s_{n-1}}V_{\varphi}(t)\\
&  =\sum_{k=1}^{n-1}U_{k}(s_{n-1})\mathbf{1}_{\left\{  s_{k-1}<\tau\leq
s_{k}\right\}  }+(U_{n}(s_{n-1})+y_{n}c(s_{n-1}))\mathbf{1}_{\left\{
s_{n-1}<\tau\right\}  }.
\end{align*}
It follows that%
\begin{align*}
U_{k}(s_{n-1})\mathbf{1}_{\left\{  s_{k-1}<\tau\leq s_{k}\right\}  }  &
=0\text{\quad for each }k=1,\ldots,n-1,\\
(U_{n}(s_{n-1})+y_{n}c(s_{n-1}))\mathbf{1}_{\left\{  s_{n-1}<\tau\right\}  }
&  =0.
\end{align*}
Because $U_{k}(s_{n-1})$ and $U_{n}(s_{n-1})+y_{n}c(s_{n-1})$ are
$\mathcal{F}_{T}$-measurable, it follows by assumption~(\ref{d7465edgd}) that
\begin{align*}
U_{k}(s_{n-1})  &  =0\text{\quad for each }k=1,\ldots,n-1,\\
U_{n}(s_{n-1})+y_{n}c(s_{n-1})  &  =0.
\end{align*}
The value of the strategy at time $s_{n}$ is%
\begin{align*}
0  &  \leq V_{\varphi}(s_{n})\\
&  =\sum_{k=1}^{n-1}U_{k}(s_{n})\mathbf{1}_{\left\{  s_{k-1}<\tau\leq
s_{k}\right\}  }+U_{n}(s_{n})\mathbf{1}_{\left\{  s_{n-1}<\tau\right\}
}+y_{n}c(s_{n})\mathbf{1}_{\left\{  s_{n}<\tau\right\}  }\\
&  =\sum_{k=1}^{n}U_{k}(s_{n})\mathbf{1}_{\left\{  s_{k-1}<\tau\leq
s_{k}\right\}  }+(U_{n}(s_{n})+y_{n}c(s_{n}))\mathbf{1}_{\left\{  s_{n}%
<\tau\right\}  }.
\end{align*}
Hence%
\begin{align*}
U_{k}(s_{n})\mathbf{1}_{\left\{  s_{k-1}<\tau\leq s_{k}\right\}  }  &
\geq0\text{\quad for each }k=1,\ldots,n,\\
(U_{n}(s_{n})+y_{n}c(s_{n}))\mathbf{1}_{\left\{  s_{n}<\tau\right\}  }  &
\geq0.
\end{align*}
Because $U_{k}(s_{n})$ and $U_{n}(s_{n})+y_{n}c(s_{n})$ are $\mathcal{F}_{T}%
$-measurable, it follows by assumption~(\ref{d7465edgd}) that%
\begin{align*}
U_{k}(s_{n})  &  \geq0\text{\quad for each }k=1,\ldots,n,\\
U_{n}(s_{n})+y_{n}c(s_{n})  &  \geq0.
\end{align*}
Since the value $U_{n}(t)$ of the admissible self-financing strategy
$x_{n}(t)$ in the Black--Scholes model is an $(\mathcal{F}_{t})_{t\geq0}%
$-martingale under~$Q_{BS}$, it follows that%
\begin{align*}
0  &  \leq\mathbb{E}_{Q_{BS}}(U_{n}(s_{n})+y_{n}c(s_{n})|\mathcal{F}_{s_{n-1}%
})=U_{n}(s_{n-1})+y_{n}\mathbb{E}_{Q_{BS}}\left(  c(s_{n})|\mathcal{F}%
_{s_{n-1}}\right) \\
&  =y_{n}\left(  \mathbb{E}_{Q_{BS}}\left(  c(s_{n})|\mathcal{F}_{s_{n-1}%
}\right)  -c(s_{n-1})\right)  .
\end{align*}
Because $\mathbb{E}_{Q_{BS}}\left(  c(s_{n})|\mathcal{F}_{s_{n-1}}\right)
>c(s_{n-1})$, we can see that $y_{n}\geq0$. On the other hand,
\[
0\leq\mathbb{E}_{Q_{BS}}(U_{n}(s_{n})|\mathcal{F}_{s_{n-1}})=U_{n}%
(s_{n-1})=-y_{n}c(s_{n-1}).
\]
Since $c(s_{n-1})>0$, it follows that $y_{n}\leq0$. Hence, we have shown that
$y_{n}=0$. As a result, $\mathbb{E}_{Q_{BS}}(U_{n}(s_{n})|\mathcal{F}%
_{s_{n-1}})=U_{n}(s_{n-1})=-y_{n}c(s_{n-1})=0$. Because $U_{n}(s_{n})\geq0$,
it follows that $U_{n}(s_{n})=0$. Moreover, for each $k=1,\ldots,n-1$, we have%
\[
\mathbb{E}_{Q_{BS}}(U_{k}(s_{n})|\mathcal{F}_{s_{n-1}})=U_{k}(s_{n-1})=0
\]
and $U_{k}(s_{n})\geq0$, which means that $U_{k}(s_{n})=0$.\ Hence,%
\[
V_{\varphi}(s_{n})=\sum_{k=1}^{n}U_{k}(s_{n})\mathbf{1}_{\left\{  s_{k-1}%
<\tau\leq s_{k}\right\}  }+(U_{n}(s_{n})+y_{n}c(s_{n}))\mathbf{1}_{\left\{
s_{n}<\tau\right\}  }=0,
\]
completing the induction step.

It remains to consider the case when $V_{\varphi}(s_{n})<0$ with positive
probability~$P$ for some $n=0,\ldots,N$. Let~$m$ be the largest integer~$n$
among $0,\ldots,N$ such that $V_{\varphi}(s_{n})<0$ with positive
probability~$P$. Clearly, $0<m<N$ since $V_{\varphi}(0)=0$ and $V_{\varphi
}(s_{N})=V_{\varphi}(T)\geq0$. Because, by (\ref{Eq:nv74oa0m}),%
\begin{align*}
0  &  \leq V_{\varphi}(s_{m+1})\\
&  =\sum_{k=1}^{m}U_{k}(s_{m+1})\mathbf{1}_{\left\{  s_{k-1}<\tau\leq
s_{k}\right\}  }+U_{m+1}(s_{m+1})\mathbf{1}_{\left\{  s_{m}<\tau\right\}
}+y_{m+1}c(s_{m+1})\mathbf{1}_{\left\{  s_{m+1}<\tau\right\}  }\\
&  =\sum_{k=1}^{m+1}U_{k}(s_{m+1})\mathbf{1}_{\left\{  s_{k-1}<\tau\leq
s_{k}\right\}  }+(U_{m+1}(s_{m+1})+y_{m+1}c(s_{m+1}))\mathbf{1}_{\left\{
s_{m+1}<\tau\right\}  },
\end{align*}
we can see that for each $k=1,\ldots,m+1$ we have $U_{k}(s_{m+1}%
)\mathbf{1}_{\left\{  s_{k-1}<\tau\leq s_{k}\right\}  }\geq0$, hence
$U_{k}(s_{m+1})\geq0$ by assumption~(\ref{d7465edgd}), which implies that
\[
0\leq\mathbb{E}_{Q_{BS}}(U_{k}(s_{m+1})|\mathcal{F}_{s_{m}})=U_{k}(s_{m}).
\]
Moreover, by (\ref{Eq:nv74oa0m}), we also have%
\begin{align*}
V_{\varphi}(s_{m})  &  =\sum_{k=1}^{m-1}U_{k}(s_{m})\mathbf{1}_{\left\{
s_{k-1}<\tau\leq s_{k}\right\}  }+U_{m}(s_{m})\mathbf{1}_{\left\{
s_{m-1}<\tau\right\}  }+y_{m}c(s_{m})\mathbf{1}_{\left\{  s_{m}<\tau\right\}
}\\
&  =\sum_{k=1}^{m}U_{k}(s_{m})\mathbf{1}_{\left\{  s_{k-1}<\tau\leq
s_{k}\right\}  }+(U_{m}(s_{m})+y_{m}c(s_{m}))\mathbf{1}_{\left\{  s_{m}%
<\tau\right\}  }.
\end{align*}
As a result,%
\[
V_{\varphi}(s_{m})\mathbf{1}_{\left\{  \tau\leq s_{m}\right\}  }=\sum
_{k=1}^{m}U_{k}(s_{m})\mathbf{1}_{\left\{  s_{k-1}<\tau\leq s_{k}\right\}
}\geq0.
\]
This implies that
\[
P\left(  V_{\varphi}(s_{m})<0\right)  =P\left(  \left\{  V_{\varphi}%
(s_{m})<0\right\}  \cap\left\{  s_{m}<\tau\right\}  \right)  .
\]
Since $\left\{  V_{\varphi}(s_{m})<0\right\}  \in\mathcal{G}_{s_{m}}$, it
follows by the properties of the enlarged filtration (for example, Lemma~5.27
in~\cite{CapZas2016}) that there is an $A\in\mathcal{F}_{s_{m}}$ such that
\[
\left\{  V_{\varphi}(s_{m})<0\right\}  \cap\left\{  s_{m}<\tau\right\}
=A\cap\left\{  s_{m}<\tau\right\}  .
\]
We define a new quasi-simple self-financing strategy $\psi=(\psi^{B},\psi
^{S},\psi^{D})$ by%
\[%
\begin{tabular}
[c]{ll}%
$\psi^{B}(t):=\psi^{S}(t):=\psi^{D}(t):=0$ & for $t\in\lbrack0,s_{m}],$\\
& \\
$\psi^{B}(t):=\left(  \varphi^{B}(t)-V_{\varphi}(s_{m})\right)  \mathbf{1}%
_{A\cap\left\{  s_{m}<\tau\right\}  },$ & \\
$\psi^{S}(t):=\varphi^{S}(t)\mathbf{1}_{A\cap\left\{  s_{m}<\tau\right\}  },$
$\psi^{D}(t):=\varphi^{D}(t)\mathbf{1}_{A\cap\left\{  s_{m}<\tau\right\}
}\quad$ & for $t\in(s_{m},T].$%
\end{tabular}
\
\]
For this strategy, we put%
\begin{align*}
\tilde{x}_{n}^{B}(t)  &  :=\mathbf{1}_{m<n}\mathbf{1}_{A}(x_{n}^{B}%
(t)-V_{\varphi}(s_{m})),\\
\tilde{x}_{n}^{S}(t)  &  :=\mathbf{1}_{m<n}\mathbf{1}_{A}x_{n}^{S}(t),\\
\tilde{y}_{n}  &  :=\mathbf{1}_{m<n}\mathbf{1}_{A}y_{n}%
\end{align*}
in place of $x_{n}^{B}(t),x_{n}^{S}(t),y_{n}$ in
Definition~\ref{Def:957dh45d1a}. Then, for $n=1,\ldots,m$ and $t\in
(s_{n-1},s_{n}]$, we have
\begin{align*}
\psi^{B}(t)  &  =\tilde{x}_{n\wedge\mu}^{B}(t)=0,\\
\psi^{S}(t)  &  =\tilde{x}_{n\wedge\mu}^{S}(t)=0,\\
\psi^{D}(t)  &  =\tilde{y}_{n\wedge\mu}=0.
\end{align*}
Moreover, for $n=m+1,\ldots,N$ and $t\in(s_{n},s_{n+1}]$, we have
\begin{align*}
\psi^{B}(t)  &  =\tilde{x}_{n\wedge\mu}^{B}(t)=\mathbf{1}_{\{m<n\wedge\mu
\}}\mathbf{1}_{A}(x_{n\wedge\mu}^{B}(t)-V_{\varphi}(s_{m}))\\
&  =\mathbf{1}_{\{m<\mu\}}\mathbf{1}_{A}(\varphi^{B}(t)-V_{\varphi}%
(s_{m}))=\mathbf{1}_{\{s_{m}<\tau\}}\mathbf{1}_{A}(\varphi^{B}(t)-V_{\varphi
}(s_{m})),\\
\psi^{S}(t)  &  =\tilde{x}_{n\wedge\mu}^{S}(t)=\mathbf{1}_{\{m<n\wedge\mu
\}}\mathbf{1}_{A}x_{n\wedge\mu}^{S}(t)\\
&  =\mathbf{1}_{\{m<\mu\}}\mathbf{1}_{A}\varphi^{S}(t)=\mathbf{1}%
_{\{s_{m}<\tau\}}\mathbf{1}_{A}\varphi^{S}(t),\\
\psi^{D}(t)  &  =\tilde{y}_{n\wedge\mu}=\mathbf{1}_{\{m<n\wedge\mu
\}}\mathbf{1}_{A}y_{n\wedge\mu}=\mathbf{1}_{\{m<\mu\}}\mathbf{1}_{A}%
\varphi^{D}(t)=\mathbf{1}_{\{s_{m}<\tau\}}\mathbf{1}_{A}\varphi^{D}(t).
\end{align*}
This means that $\psi$ is indeed a quasi-simple self-financing strategy. The
value of this strategy is%
\[%
\begin{tabular}
[c]{ll}%
$V_{\psi}(t)=0$ & $\text{for }t\in\lbrack0,s_{m}],$\\
$V_{\psi}(t)=\left(  V_{\varphi}(t)-V_{\varphi}(s_{m})\right)  \mathbf{1}%
_{A\cap\left\{  s_{m}<\tau\right\}  }\text{\quad}$ & $\text{for }t\in
(s_{m},T].$%
\end{tabular}
\
\]
Since $V_{\varphi}(s_{n})\geq0$ for each $n>m$ and $V_{\varphi}(s_{m})<0$ on
$A\cap\left\{  s_{m}<\tau\right\}  $, we have $V_{\psi}(s_{n})\geq0$ for
all~$n=0,1,\ldots,N$, that is, $\psi$~belongs to the class of strategies
considered earlier. Hence we know that $V_{\psi}(T)=0$. On the other hand,
$V_{\psi}(T)=\left(  V_{\varphi}(T)-V_{\varphi}(s_{m})\right)  >0$ on
$A\cap\left\{  s_{m}<\tau\right\}  $, so
\[
0=P\left(  A\cap\left\{  s_{m}<\tau\right\}  \right)  =P\left(  \left\{
V_{\varphi}(s_{m})<0\right\}  \cap\left\{  s_{m}<\tau\right\}  \right)
=P\left(  V_{\varphi}(s_{m})<0\right)  .
\]
This contradicts the definition of~$m$ and completes the proof of the theorem.
\end{proof}

\section{Survival and hazard processes}

In the literature, for example \cite{BieRut2002}, \cite{BieJeaRut2009}, where
the existence of a measure~$Q$ turning both $e^{-rt}S(t)$ and $e^{-rt}D(t,T)$
into $(\mathcal{G}_{t})_{t\in[0,T]}$-martingales is stipulated, the survival
process\ $G(t)$ and hazard process~$\Gamma(t)$ are defined in terms of~$Q$ as
\begin{equation}
G(t)=e^{-\Gamma(t)}=Q(t<\tau|\mathcal{F}_{T}). \label{s64hbfkapenk}%
\end{equation}

However, no such measure~$Q$ is given here, as our goal is to construct it
starting from minimalist assumptions. Because of this, we proceed in a
different fashion to define the survival process (hence also the hazard
process). The key is the result that the strictly positive process
$e^{-rt}c(t)$ is a strict submartingale with respect to the Black--Scholes
filtration $(\mathcal{F}_{t})_{t\in[0,T]}$ and measure~$Q_{BS}$, as shown in
Theorem~\ref{Thm:d64gsna90ms}. Therefore, we can apply the Doob--Meyer type
multiplicative decomposition for positive submartingales; see Yoeurp and
Meyer \cite{YoeMey1976} and Azema \cite{Azema1978}. This will lead to a
definition of the survival and hazard processes that does not involve~$Q$.

\begin{theorem}
\label{Thm:mvja83fab}Under the assumptions in Section~\ref{Sect:f745hhandta},
the following conditions are equivalent:

\begin{enumerate}
\item The process $e^{-rt}c(t)$ is a strict $\left(  \mathcal{F}_{t}\right)
_{t\in[0,T]}$-submartingale under~$Q_{BS}$;

\item \label{Thm:mvja83fab_ii}There is a strictly positive non-increasing
strict $(\mathcal{F}_{t})_{t\in\lbrack0,T]}$-supermartingale $G(t)$ under
$Q_{BS}$ such that $G(0)=1$, $G(t)$ has continuous paths, and $e^{-rt}%
c(t)G(t)$ is an $(\mathcal{F}_{t})_{t\in\lbrack0,T]}$-martingale under
$Q_{BS}$.
\end{enumerate}

\noindent Moreover, if such a process $G(t)$ exists, it is unique up to indistinguishability.
\end{theorem}

\begin{proof}
Without loss of generality, we can take $r=0$ to simplify the proof.

Suppose that $c(t)$ is an $\left(  \mathcal{F}_{t}\right)  _{t\in\lbrack0,T]}%
$-submartingale under~$Q_{BS}$. Because $c(t)$ is bounded, strictly positive,
with continuous paths and satisfies $c(T)=1$, all its paths are bounded away
from~$0$ on the closed interval $[0,T]$. Moreover, the augmented filtration
$(\mathcal{F}_{t})_{t\in\lbrack0,T]}$ generated by Brownian motion satisfies
the usual conditions. Therefore, Theorem~33 in \cite{Azema1978} (or Theorem~2
of \cite{YoeMey1976}) gives a unique multiplicative decomposition of $c(t)$ on
$[0,T]$. Namely, there is a unique (up to indistinguishability) strictly
positive non-increasing $(\mathcal{F}_{t})_{t\in\lbrack0,T]}$-previsible
process $G(t)$ on $[0,T]$ such that $G(0)=1$ and $c(t)G(t)$ is an
$(\mathcal{F}_{t})_{t\in\lbrack0,T]}$-martingale under~$Q_{BS}$. Because
$c(T)=1$, it follows that for each $t\in\lbrack0,T]$%
\[
c(t)G(t)=\mathbb{E}_{Q_{BS}}(G(T)|\mathcal{F}_{t}).
\]
Since $G(0)=1$ and $G(t)$ is positive and non-increasing, it follows that
$G(T)$ is bounded, so it is square integrable. Hence, $c(t)G(t)=\mathbb{E}%
_{Q_{BS}}(G(T)|\mathcal{F}_{t})$ has continuous paths by the martingale
representation theorem. It follows that $G(t)$ also has continuous paths given
that $c(t)$ does and is positive.

It remains to verify that, if $G(t)$ is a non-increasing process and
$c(t)G(t)$ is a martingale, then $c(t)$ is a strict submartingale if and only
if $G(t)$ is a strict supermartingale. Let $0\leq s<t\leq T$. Because
$\mathbb{E}_{Q_{BS}}(c(t)G(t)|\mathcal{F}_{s})=c(s)G(s)$, the inequality%
\[
c(s)<\mathbb{E}_{Q_{BS}}(c(t)|\mathcal{F}_{s})
\]
holds if and only if%
\[
\mathbb{E}_{Q_{BS}}(c(t)G(t)|\mathcal{F}_{s})<\mathbb{E}_{Q_{BS}%
}(c(t)G(s)|\mathcal{F}_{s}),
\]
that is, if and only if%
\[
\mathbb{E}_{Q_{BS}}(c(t)G(t)\mathbf{1}_{A})<\mathbb{E}_{Q_{BS}}%
(c(t)G(s)\mathbf{1}_{A})
\]
for every $A\in\mathcal{F}_{s}$ of positive measure. Because $G(t)\leq G(s)$
and $c(t)$ is strictly positive, this is so if and only if%
\[
\left\{  c(t)G(t)\mathbf{1}_{A}<c(t)G(s)\mathbf{1}_{A}\right\}  =\left\{
G(t)\mathbf{1}_{A}<G(s)\mathbf{1}_{A}\right\}
\]
is a set of positive measure for every $A\in\mathcal{F}_{s}$ of positive
measure. Using, the inequality $G(t)\leq G(s)$ once again, we can see that
this is, in turn, equivalent to%
\[
\mathbb{E}_{Q_{BS}}(G(t)|\mathcal{F}_{s})<G(s),
\]
completing the argument.

It remains to show the uniqueness of the process~$G(t)$ whose existence is
asserted in~\ref{Thm:mvja83fab_ii}. Indeed, since every left-continuous
adapted process is previsible, it follows that~$G(t)$ is previsible. The
uniqueness (up to indistinguishability) of the multiplicative decomposition of
the submartingale~$c(t)$ therefore gives that of the process~$G(t)$.
\end{proof}

\begin{definition}
\label{Def:f740aqpmno}\upshape We call the unique process $G(t)$ in
Theorem~\ref{Thm:mvja83fab} the \emph{survival process} and $\Gamma(t):=-\log
G(t)$ the \emph{hazard process}.
\end{definition}

In the next section we use $G(t)$ (or, equivalently, $\Gamma(t)$) to construct
a measure~$Q$ extending $Q_{BS}$ from the $\sigma$-algebra $\mathcal{F}_{T}$
to~$\mathcal{G}_{T}$ such that $e^{-rt}S(t)$ and $e^{-rt}D(t,T)$ become
$\{\mathcal{G}(t)\}_{t\in[0,T]}$-martingales under~$Q$. To achieve this we
construct~$Q$ in such a way that $G(t)$ and $\Gamma(t)$ will be expressed in
terms of~$Q$ as in~(\ref{s64hbfkapenk}). This justifies calling them the
\emph{survival process} and \emph{hazard process} in
Definition~\ref{Def:f740aqpmno}.

\section{\label{Sect:hf64bs9a}Construction of martingale measure}

A probability measure $Q$ extending $Q_{BS}$ from $\mathcal{F}_{T}$
to~$\mathcal{G}_{T}$ such that $G(t)$ is given by~(\ref{s64hbfkapenk}) would
need to satisfy
\begin{align}
Q(A\cap\left\{  s<\tau\leq t\right\}  )  &  =\mathbb{E}_{Q_{BS}}%
(\mathbf{1}_{A}(G(s)-G(t))),\label{fbhr6s63ga1}\\
Q(A\cap\{T<\tau\})  &  =\mathbb{E}_{Q_{BS}}(\mathbf{1}_{A}G(T))
\label{fbhr6s63ga2}%
\end{align}
for each $s,t\in\lbrack0,T]$ such that $s<t$ and each $A\in\mathcal{F}_{T}$.

With the aim of constructing such a measure~$Q$, we first construct a
measure~$\tilde{Q}$ on the space%
\[
\tilde{\Omega}:=\dot{\mathbb{R}}^{[0,T]}\times(0,\infty),
\]
where $\dot{\mathbb{R}}:=\mathbb{R}\cup\{\infty\}$ is the one-point
compactification of~$\mathbb{R}$, and then pull it back to $\Omega$ by the map
$(W,\tau):\Omega\rightarrow\tilde{\Omega}$ to obtain~$Q$. Here $\dot
{\mathbb{R}}^{[0,T]}$ can be regarded as the space of paths of the Brownian
motion~$W$ driving the Black--Scholes model, and $(0,\infty)$ as the space of
values of~$\tau$.

If the pulled-back measure $Q$ is to satisfy (\ref{fbhr6s63ga1})
and~(\ref{fbhr6s63ga2}), then we need to put%
\begin{align}
\tilde{Q}(A_{z,Z}\times(s,t])  &  :=\mathbb{E}_{Q_{BS}}(\mathbf{1}_{\{W\in
A_{z,Z}\}}(G(t)-G(s))),\label{p0w75nslnd21}\\
\tilde{Q}(A_{z,Z}\times(T,\infty))  &  :=\mathbb{E}_{Q_{BS}}(\mathbf{1}%
_{\{W\in A_{z,Z}\}}G(T)) \label{p0w75nslnd22}%
\end{align}
for each $s,t\in\lbrack0,T]$ such that $s<t$ and each \emph{cylindrical
set}$~A_{z,Z}$ in~$\dot{\mathbb{R}}^{[0,T]}$ of the form
\[
A_{z,Z}:=\{x\in\dot{\mathbb{R}}^{[0,T]}:(x(z_{1}),\ldots,x(z_{n}))\in Z\}
\]
for some non-negative integer~$n$, some $z=(z_{1},\ldots,z_{n})\in
\lbrack0,T]^{n}$ and some Borel set $Z\in\mathcal{B}(\dot{\mathbb{R}}^{n})$.

Next, $\tilde{Q}$ can be extended in the standard manner to an additive set
function on the algebra~$\tilde{\mathcal{A}}$ consisting of finite unions of
disjoint sets of the form $A_{z,Z}\times(s,t]$ or $A_{z,Z}\times(T,\infty)$,
where $s,t\in\lbrack0,T]$ with $s<t$, and where $A_{z,Z}$ is a cylindrical set
in~$\dot{\mathbb{R}}^{[0,T]}$. The next step is to show that $\tilde{Q}$ is a
measure (that is, it is countably additive) on the algebra~$\tilde
{\mathcal{A}}$.

\begin{lemma}
\label{Lem:hf649sja}$\tilde{Q}$ is a countably additive set function
on~$\tilde{\mathcal{A}}$.
\end{lemma}

\begin{proof}
The following argument resembles in some respects the standard proof of
countable additivity of Wiener measure on path space.

Let $\mathcal{K}$ be the class of subsets in $\tilde{\Omega}$ of the form
$K_{u,U}\cup(L_{v,V}\times(T,\infty))$, where%
\begin{align*}
L_{u,U}  &  :=\{(x,t)\in\dot{\mathbb{R}}^{[0,T]}\times(0,T]:(x(u_{1}%
),\ldots,x(u_{k}),t)\in U\},\\
M_{v,V}  &  :=\{x\in\dot{\mathbb{R}}^{[0,T]}:(x(v_{1}),\ldots,x(v_{l}))\in
V\},
\end{align*}
for some positive integers \thinspace$l,m$, some $u=(u_{1},\ldots,u_{l}%
)\in\lbrack0,T]^{l}$, $v=(v_{1},\ldots,v_{m})\in\lbrack0,T]^{m}$ and some
compact sets $U\subset\mathbb{R}^{l}\times(0,T]$, $V\subset\mathbb{R}^{m}$.

We claim that $\mathcal{K}$ is a compact class in~$\tilde{\Omega}$ (see
Definition~1.4.1 in \cite{Bog2007}). To prove this, take any sequence of sets
$K_{n}\in\mathcal{K}$ such that $\bigcap_{n=1}^{\infty}K_{n}=\emptyset$. Then
$K_{n}=L_{u_{n},U_{n}}\cup(M_{v_{n},V_{n}}\times(T,\infty))$ for some positive
integers $l_{n},m_{n}$, some $u_{n}\in\lbrack0,T]^{l_{n}}$, $v_{n}\in
\lbrack0,T]^{m_{n}}$ and some compact sets~$U_{n}\subset\mathbb{R}^{l_{n}%
}\times(0,T]$, $V_{n}\subset\mathbb{R}^{m_{n}}$. We can see
(c.f.~Theorem~A5.17 in~\cite{Ash1972}) that $U_{n},V_{n}$ are closed subsets
in $\dot{\mathbb{R}}^{l_{n}}\times\lbrack0,T]$ and, respectively,~$\dot
{\mathbb{R}}^{m_{n}}$. Hence $L_{u_{n},U_{n}},M_{v_{n},V_{n}}$ are closed in
the product topology in $\dot{\mathbb{R}}^{[0,T]}\times\lbrack0,T]$ and,
respectively,~$\dot{\mathbb{R}}^{[0,T]}$. It also follows that $\bigcap
_{n=1}^{\infty}L_{u_{n},U_{n}}=\emptyset$ and $\bigcap_{n=1}^{\infty}%
M_{v_{n},V_{n}}=\emptyset$. By the Tikhonov theorem (the product of any family
of compact sets is compact in the product topology), $\dot{\mathbb{R}}%
^{[0,T]}\times\lbrack0,T]$ and $\dot{\mathbb{R}}^{[0,T]}$ are compact.
Therefore, there is an~$N$ such that $\bigcap_{n=1}^{N}L_{u_{n},U_{n}%
}=\emptyset$ and $\bigcap_{n=1}^{N}M_{v_{n},V_{n}}=\emptyset$, hence
$\bigcap_{n=1}^{N}K_{n}=\emptyset$. This proves that~$\mathcal{K}$ is a
compact class.

We also claim that~$\mathcal{K}$ is a class approximating the additive set
function~$\tilde{Q}$ on the algebra~$\tilde{\mathcal{A}}$, that is, for any
$A\in\tilde{\mathcal{A}}$ and $\varepsilon>0$ there exist $K_{\varepsilon}%
\in\mathcal{K}$ and $A_{\varepsilon}\in\tilde{\mathcal{A}}$ such that
$A_{\varepsilon}\subset K_{\varepsilon}\subset A$ and $\tilde{Q}(A\setminus
A_{\varepsilon})<\varepsilon$. (This definition of an approximating class
comes from Theorem~1.4.3 in~\cite{Bog2007}.) Take any $A\in\tilde{\mathcal{A}%
}$ and $\varepsilon>0$. We can write $A$ as
\begin{equation}
A=\bigcup_{i=1}^{m}(A_{z,Z_{i}}\times(s_{i},t_{i}])\cup(A_{z,Z_{m+1}}%
\times(T,\infty)) \label{fgrtdcs74jj9}%
\end{equation}
for some non-negative integer~$m$, some $s_{1},t_{1},\ldots,s_{m},t_{m}%
\in\lbrack0,T]$ such that%
\[
0\leq s_{1}<t_{1}\leq\cdots\leq s_{m}<t_{m}\leq T
\]
and some cylindrical sets $A_{z,Z_{1}},\ldots,A_{z,Z_{m}}$ and $A_{z,Z_{m+1}}$
in~$\dot{\mathbb{R}}^{[0,T]}$. In particular, note that the cylindrical sets
can be chosen to share the same tuple $z\in\mathbb{R}^{n}$ for some
non-negative integer~$n$, with $Z_{1},\ldots,Z_{m+1}\in\mathcal{B}%
(\dot{\mathbb{R}}^{n})$.

Let $\eta:=\frac{\varepsilon}{2m+1}$. Regularity of the Borel sets
$Z_{1},\ldots,Z_{m+1}$ (see Theorem~1.4.8 in~\cite{Bog2007}) implies that
there are compact sets $F_{1},\ldots,F_{m+1}\subset\mathbb{R}^{n}$ such that
for each $i=1,\ldots,m+1$ we have $F_{i}\subset Z_{i}$ and%
\[
Q_{BS}\{W\in A_{z,A_{i}\setminus F_{i}}\}<\eta.
\]
Moreover, since $G$ has non-increasing continuous paths, it follows that for
each $i=1,\ldots,m$ we have $G(s_{i})-G(t)\searrow0$ as $t\searrow s_{i}$. By
monotone convergence, it follows that $\mathbb{E}_{Q_{BS}}(G(s_{i}%
)-G(t))\searrow0$ as $t\searrow s_{i}$. Hence there is a $v_{i}\in(s_{i}%
,t_{i})$ such that%
\[
\mathbb{E}_{Q_{BS}}(G(s_{i})-G(v_{i}))<\eta.
\]
Next, we take some $w_{i}\in(s_{i},v_{i})$ for each $i=1,\ldots,m$ and put%
\begin{align*}
K_{\varepsilon}  &  :=\bigcup_{i=1}^{m}(A_{z,F_{i}}\times\lbrack w_{i}%
,t_{i}])\cup(A_{z,F_{m+1}}\times(T,\infty)),\\
A_{\varepsilon}  &  :=\bigcup_{i=1}^{m}(A_{z,F_{i}}\times(v_{i},t_{i}%
])\cup(A_{z,F_{m+1}}\times(T,\infty)).
\end{align*}
Clearly, $A_{\varepsilon}\in\tilde{\mathcal{A}}$, $K_{\varepsilon}%
\in\mathcal{K}$ and $A_{\varepsilon}\subset K_{\varepsilon}\subset A$. Since%
\begin{align*}
&  A\setminus A_{\varepsilon}\\
&  =\bigcup_{i=1}^{m}\left(  A_{z,Z_{i}}\times(s_{i},v_{i}]\right)
\cup\bigcup_{i=1}^{m}\left(  A_{z,Z_{i}\setminus F_{i}}\times(v_{i}%
,t_{i}]\right)  \cup(A_{z,Z_{m+1}\setminus F_{m+1}}\times(T,\infty)),
\end{align*}
which is a union of disjoint sets, it follows that $\tilde{Q}(A\setminus
A_{\varepsilon})$ is the sum of the following three terms:%
\begin{align*}
\sum_{i=1}^{m}\tilde{Q}\left(  A_{z,Z_{i}}\times(s_{i},v_{i}]\right)   &
=\sum_{i=1}^{m}\mathbb{E}_{Q_{BS}}(1_{\{W\in A_{z,Z_{i}}\}}(G(s_{i}%
)-G(v_{i})))\\
&  \leq\sum_{i=1}^{m}\mathbb{E}_{Q_{BS}}\left[  G(s_{i})-G(v_{i})\right]
<m\eta,\\
\sum_{i=1}^{m}\tilde{Q}\left(  A_{z,Z_{i}\setminus F_{i}}\times(v_{i}%
,t_{i}]\right)   &  =\sum_{i=1}^{m}\mathbb{E}_{Q_{BS}}(1_{\{W\in
A_{z,Z_{i}\setminus F_{i}}\}}(G(v_{i})-G(t_{i})))\\
&  \leq\sum_{i=1}^{m}Q_{BS}(W\in A_{z,Z_{i}\setminus F_{i}})\leq m\eta,\\
\tilde{Q}(A_{z,Z_{m+1}\setminus F_{m+1}}\times(T,\infty))  &  =\mathbb{E}%
_{Q_{BS}}(1_{\{W\in A_{z,Z_{m+1}\setminus F_{m+1}}\}}G(T))\\
&  \leq Q_{BS}(W\in A_{z,Z_{m+1}\setminus F_{m+1}})<\eta.
\end{align*}
Hence $\tilde{Q}(A\setminus A_{\varepsilon})<m\eta+m\eta+\eta=(2m+1)\eta
<\varepsilon$, which shows that~$\mathcal{K}$ is a class approximating~$\tilde
{Q}$ on the algebra~$\tilde{\mathcal{A}}$. By Theorem~1.4.3 in~\cite{Bog2007},
it follows that~$\tilde{Q}$ is a countably additive set function
on~$\tilde{\mathcal{A}}$, completing the proof.
\end{proof}

\medskip

Next, we introduce the algebra~$\mathcal{A}$ of subsets of~$\Omega$ of the
form $\{(W,\tau)\in A\}$ such that $A\in\tilde{\mathcal{A}}$, and define a set
function~$Q$ on~$\mathcal{A}$ by%
\begin{equation}
Q((W,\tau)\in A):=\tilde{Q}(A) \label{ns74ma0r4nk}%
\end{equation}
for each $A\in\tilde{\mathcal{A}}$. The following lemma is needed to show that
$Q$ is well defined and countably additive on~$\mathcal{A}$.

\begin{lemma}
\label{Lem:mnf6sw6fdgbv}Let $A\in\tilde{\mathcal{A}}$ be such that $\left\{
(W,\tau)\in A\right\}  =\emptyset$. Then $\tilde{Q}(A)=0$.
\end{lemma}

\begin{proof}
We can write $A$ as in (\ref{fgrtdcs74jj9}). Then%
\[
\left\{  (W,\tau)\in A\right\}  =\bigcup_{i=1}^{m}(\{W\in A_{z,Z_{i}}%
\}\cap\{s_{i}<\tau\leq t_{i}\})\cup(\{W\in A_{z,Z_{m+1}}\}\cap\{T<\tau\}).
\]
Since $\left\{  (W,\tau)\in A\right\}  =\emptyset$, it follows that $\{W\in
A_{z,Z_{i}}\}\cap\{s_{i}<\tau\leq t_{i}\}=\emptyset$ for each $i=1,\ldots,m$
and $\{W\in A_{z,Z_{m+1}}\}\cap\{T<\tau\}=\emptyset$.
Assumption~(\ref{d7465edgd}) implies that $P(W\in A_{z,Z_{i}})=0$, hence
$Q_{BS}(W\in A_{z,Z_{i}})=0$ for each $i=1,\ldots,m+1$. As a result,%
\begin{align*}
\tilde{Q}(A)  &  =\sum_{i=1}^{m}\tilde{Q}(A_{z,Z_{i}}\times(s_{i}%
,t_{i}])+\tilde{Q}(A_{u,Zm+1}\times(T,\infty))\\
&  =\sum_{i=1}^{m}\mathbb{E}_{Q_{BS}}(1_{\{W\in A_{z,Z_{i}}\}}(G(s_{i}%
)-G(t_{i})))+\mathbb{E}_{Q_{BS}}(1_{\{W\in A_{z,Z_{m+1}}\}}G(T))=0,
\end{align*}
completing the proof.
\end{proof}

\begin{proposition}
The set function $Q$ is well defined on~$\mathcal{A}$ by $(\ref{ns74ma0r4nk}%
)$, that is, if $\{(W,\tau)\in A\}=\{(W,\tau)\in B\}$ for some $A,B\in
\tilde{\mathcal{A}}$, then $\tilde{Q}(A)=\tilde{Q}(B)$.
\end{proposition}

\begin{proof}
If $\{(W,\tau)\in A\}=\{(W,\tau)\in B\}$, then $\{(W,\tau)\in A\vartriangle
B\}=\emptyset$, where $A\vartriangle B=(A\setminus B)\cup(B\setminus A)$
denotes the symmetric difference. By Lemma~\ref{Lem:mnf6sw6fdgbv}, it follows
that $\tilde{Q}(A\vartriangle B)=0$, hence $\tilde{Q}(A)=\tilde{Q}(B)$.
\end{proof}

\begin{proposition}
\label{Cor:nfrt6434ad}$Q$ is a countably additive function on the
algebra~$\mathcal{A}$.
\end{proposition}

\begin{proof}
Let $B_{n}\in\mathcal{A}$ be a sequence of disjoint sets such that
$\bigcup_{n=1}^{\infty}B_{n}\in\mathcal{A}$. For each~$n$ we can write
$B_{n}=\{(W,\tau)\in A_{n}\}$ for some $A_{n}\in\tilde{\mathcal{A}}$. For any
$n\neq m$, since $\{(W,\tau)\in A_{n}\cap A_{m}\}=B_{n}\cap B_{m}=\emptyset$,
it follows by Lemma~\ref{Lem:mnf6sw6fdgbv} that $\tilde{Q}(A_{n}\cap A_{m}%
)=0$. We put
\[
D_{n}:=A_{n}\setminus(\left(  A_{1}\cap A_{n}\right)  \cup\cdots\cup\left(
A_{n-1}\cap A_{n}\right)  ),
\]
so that $\tilde{Q}(A_{n})=\tilde{Q}(D_{n})$ for each~$n$. The sets $D_{n}$ are
pairwise disjoint and $\bigcup_{n=1}^{\infty}A_{n}=\bigcup_{n=1}^{\infty}%
D_{n}$. By the countable additivity of $\tilde{Q}$ in~$\tilde{\mathcal{A}}$
(see Lemma~\ref{Lem:hf649sja}), it follows that%
\begin{align*}
Q\left(  \bigcup_{n=1}^{\infty}B_{n}\right)   &  =\tilde{Q}\left(
\bigcup_{n=1}^{\infty}A_{n}\right)  =\tilde{Q}\left(  \bigcup_{n=1}^{\infty
}D_{n}\right) \\
&  =\sum_{n=0}^{\infty}\tilde{Q}(D_{n})=\sum_{n=0}^{\infty}\tilde{Q}%
(A_{n})=\sum_{n=0}^{\infty}Q(B_{n}),
\end{align*}
proving countable additivity of $Q$ on~$\mathcal{A}$.\medskip
\end{proof}

Next, the Lebesgue's extension of measures (see Theorem~1.5.6
in~\cite{Bog2007}) applied to the non-negative countably additive measure~$Q$
on the algebra~$\mathcal{A}$ gives a non-negative measure, denoted by the same
symbol~$Q$, on the $\sigma$-algebra\ $\sigma(\mathcal{A})$ generated
by~$\mathcal{A}$. The final step is to extend~$Q$ to $\mathcal{G}_{T}%
=\sigma(\mathcal{F}_{T}\cup\mathcal{I}_{T})$. We put%
\begin{align*}
\mathcal{H}_{T} &  :=\sigma(W_{t},t\in\lbrack0,T]),\\
\mathcal{N}_{T} &  :=\{A\in\Sigma:A\subset B\text{ for some }B\in
\mathcal{H}_{T}\text{ such that }P(B)=0\}.
\end{align*}
Then%
\[
\mathcal{F}_{T}=\sigma(\mathcal{H}_{T}\cup\mathcal{N}_{T}),\quad
\mathcal{G}_{T}=\sigma(\mathcal{H}_{T}\cup\mathcal{N}_{T}\cup\mathcal{I}%
_{T}),\quad\sigma(\mathcal{A})=\sigma(\mathcal{H}_{T}\cup\mathcal{I}_{T}).
\]
Observe that%
\[
\mathcal{G}_{T}=\{A\in\Sigma:A\vartriangle B\in\mathcal{N}_{T}\text{ for some
}B\in\sigma(\mathcal{A})\}.
\]
Now, for any $A\in\mathcal{G}_{T}$, we put%
\[
Q(A):=Q(B)
\]
for any $B\in\sigma(\mathcal{A})$\ such that $A\vartriangle B\in
\mathcal{N}_{T}$. This does not depend on the choice of such~$B$ and defines a
non-negative measure~$Q$ on the $\sigma$-algebra~$\mathcal{G}_{T}$. It is a
probability measure since%
\[
Q(\Omega)=Q(0<\tau)=Q_{BS}(G(0))=1.
\]

In the next three propositions we show that the probability measure~$Q$
constructed above has the desired properties, namely:

\begin{itemize}
\item $Q$ coincides with the Black--Scholes risk neutral measure~$Q_{BS}$
on\ the $\sigma$-algebra~$\mathcal{F}_{T}$;

\item $Q$ satisfies~(\ref{s64hbfkapenk});

\item the discounted stock price and defaultable bond price processes
$e^{-rt}S(t)$ and $e^{-rt}D(t,T)$ are $(\mathcal{G}_{t})_{t\in\lbrack0,T]}%
$-martingales under~$Q$.
\end{itemize}

\begin{proposition}
\label{Prop:nf64tajsh}$Q=Q_{BS}$ on $\mathcal{F}_{T}$.
\end{proposition}

\begin{proof}
Because the family of sets of the form $\left\{  W\in A_{z,Z}\right\}  $,
where $A_{z,Z}$ is a cylindrical set in $\mathbb{R}^{[0,T]}$, is closed under
finite intersections and generates the $\sigma$-algebra~$\mathcal{H}_{T}$, it
suffices to show (see Lemma~1.9.4 in~\cite{Bog2007}) that%
\[
Q(W\in A_{z,Z})=Q_{BS}(W\in A_{z,Z})
\]
for any such set to prove that $Q=Q_{BS}$ on $\mathcal{H}_{T}$. Indeed, this
equality holds since%
\begin{align*}
Q(W\in A_{z,Z})  &  =Q((W,\tau)\in A_{z,Z}\times(0,\infty))\\
&  =Q((W,\tau)\in A_{z,Z}\times(0,T])+Q((W,\tau)\in A_{z,Z}\times(T,\infty))\\
&  =\tilde{Q}(A_{z,Z}\times(0,T])+\tilde{Q}(A_{z,Z}\times(T,\infty))\\
&  =\mathbb{E}_{Q_{BS}}(\mathbf{1}_{\{W\in A_{z,Z}\}}(G(0)-G(T)))+\mathbb{E}%
_{Q_{BS}}(\mathbf{1}_{\{W\in A_{z,Z}\}}G(T))\\
&  =\mathbb{E}_{Q_{BS}}(\mathbf{1}_{\{W\in A_{z,Z}\}})=Q_{BS}(W\in A_{z,Z}),
\end{align*}
where we have used the fact that $G(0)=1$. Augmenting by the null sets from
$\mathcal{N}_{T}$ preserves the equality $Q=Q_{BS}$, which therefore also
holds on $\mathcal{F}_{T}=\sigma(\mathcal{H}_{T}\cup\mathcal{N}_{T})$.
\end{proof}

\begin{proposition}
\label{Prop:ma9f65nak}For each $t\in\lbrack0,T]$,%
\[
G(t)=Q(t<\tau|\mathcal{F}_{T}).
\]

\end{proposition}

\begin{proof}
We need to show that, for each $A\in\mathcal{F}_{T}$ and $t\in\lbrack0,T]$,%
\[
\mathbb{E}_{Q}(1_{A}G(t))=Q(A\cap\{t<\tau\}).
\]
In fact, it suffices to show this equality for any $A\in\mathcal{H}_{T}$.
Because the family of sets of the form $\left\{  W\in A_{z,Z}\right\}  $,
where $A_{z,Z}$ is a cylindrical set in $\mathbb{R}^{[0,T]}$, is closed under
finite intersections and generates the $\sigma$-algebra~$\mathcal{H}_{T}$, it
suffices to show (see Lemma~1.9.4 in~\cite{Bog2007}) that%
\[
\mathbb{E}_{Q}(\mathbf{1}_{\{W\in A_{z,Z}\}}G(t))=Q(\{W\in A_{z,Z}%
\}\cap\{t<\tau\})
\]
for any such set. Indeed, since $\mathbf{1}_{\{W\in A_{z,Z}\}}G(t)$ is an
$\mathcal{F}_{T}$-measurable random variable and, by
Proposition~\ref{Prop:nf64tajsh}, $Q=Q_{BS}$ on~$\mathcal{F}_{T}$, it follows
that%
\begin{align*}
&  \mathbb{E}_{Q}(\mathbf{1}_{\{W\in A_{z,Z}\}}G(t))=\mathbb{E}_{Q_{BS}%
}(\mathbf{1}_{\{W\in A_{z,Z}\}}G(t))\\
&  \quad\quad=\mathbb{E}_{Q_{BS}}(\mathbf{1}_{\{W\in A_{z,Z}\}}%
(G(t)-G(T)))+\mathbb{E}_{Q_{BS}}(\mathbf{1}_{\{W\in A_{z,Z}\}}G(T))\\
&  \quad\quad=\tilde{Q}(A_{z,Z}\times(t,T])+\tilde{Q}(A_{z,Z}\times
(T,\infty))\\
&  \quad\quad=\tilde{Q}(A_{z,Z}\times(t,\infty))=Q(\{(W,\tau)\in A_{z,Z}%
\times(t,\infty)\})\\
&  \quad\quad=Q(\{W\in A_{z,Z}\}\cap\{t<\tau\}),
\end{align*}
as required.
\end{proof}

\begin{proposition}
Both the discounted stock price process $e^{-rt}S(t)$ and discounted
defaultable bond price process $e^{-rt}D(t,T)$ are $(\mathcal{G}_{t}%
)_{t\in\lbrack0,T]}$-martingales under~$Q$.
\end{proposition}

\begin{proof}
Without loss of generality, we can assume that $r=0$. Because the processes
$S(t)$ and $c(t)G(t)$ are $(\mathcal{F}_{t})_{t\in\lbrack0,T]}$-martingales
under~$Q_{BS}$ and $Q_{BS}=Q$ on~$\mathcal{F}_{T}$, they are also
$(\mathcal{F}_{t})_{t\in\lbrack0,T]}$-martingales under~$Q$. Thus, by
Proposition~\ref{Prop:ma9f65nak}, for any $s,t\in\lbrack0,T]$ such that $s\leq
t$ and any $A\in\mathcal{F}_{t}$, we have%
\begin{align*}
\mathbb{E}_{Q}(\mathbf{1}_{A\cap\left\{  s<\tau\right\}  }S(t))  &
=\mathbb{E}_{Q}(\mathbf{1}_{A}\mathbf{1}_{\left\{  s<\tau\right\}
}S(t))=\mathbb{E}_{Q}(\mathbf{1}_{A}\mathbb{E}_{Q}(\mathbf{1}_{\left\{
s<\tau\right\}  }|\mathcal{F}_{T})S(t))\\
&  =\mathbb{E}_{Q}(\mathbf{1}_{A}G(s)S(t))=\mathbb{E}_{Q}(\mathbf{1}%
_{A}G(s)\mathbb{E}_{Q}(S(T)|\mathcal{F}_{t}))\\
&  =\mathbb{E}_{Q}(\mathbf{1}_{A}G(s)S(T))=\mathbb{E}_{Q}(\mathbf{1}%
_{A}\mathbb{E}_{Q}(\mathbf{1}_{\left\{  s<\tau\right\}  }|\mathcal{F}%
_{T})S(T))\\
&  =\mathbb{E}_{Q}(\mathbf{1}_{A\cap\left\{  s<\tau\right\}  }S(T))
\end{align*}
and, since $D(t,T)=c(t)\mathbf{1}_{\left\{  t<\tau\right\}  }$ and $c(T)=1$,
we also have%
\begin{align*}
\mathbb{E}_{Q}(\mathbf{1}_{A\cap\left\{  s<\tau\right\}  }D(t,T))  &
=\mathbb{E}_{Q}(\mathbf{1}_{A}c(t)\mathbf{1}_{\left\{  t<\tau\right\}
})=\mathbb{E}_{Q}(\mathbf{1}_{A}c(t)\mathbb{E}_{Q}(\mathbf{1}_{\left\{
t<\tau\right\}  }|\mathcal{F}_{T}))\\
&  =\mathbb{E}_{Q}(\mathbf{1}_{A}c(t)G(t))=\mathbb{E}_{Q}(\mathbf{1}%
_{A}\mathbb{E}_{Q}(c(T)G(T)|\mathcal{F}_{t}))\\
&  =\mathbb{E}_{Q}(\mathbf{1}_{A}c(T)G(T))=\mathbb{E}_{Q}(\mathbf{1}%
_{A}G(T))=Q(A\cap\left\{  T<\tau\right\}  )\\
&  =\mathbb{E}_{Q}(\mathbf{1}_{A\cap\left\{  s<\tau\right\}  }\mathbf{1}%
_{\left\{  T<\tau\right\}  })=\mathbb{E}_{Q}(\mathbf{1}_{A\cap\left\{
s<\tau\right\}  }D(T,T)).
\end{align*}
Because the class of sets $A\cap\left\{  s<\tau\right\}  $, where
$A\in\mathcal{F}_{t}$ and $s\in\lbrack0,t]$, is closed under finite
intersections and generates the $\sigma$-algebra~$\mathcal{G}_{t}$, it follows
(see Lemma~1.9.4 in~\cite{Bog2007}) that
\[
\mathbb{E}_{Q}(\mathbf{1}_{C}S(t))=\mathbb{E}_{Q}(\mathbf{1}_{C}%
S(T))\quad\text{and}\quad\mathbb{E}_{Q}(\mathbf{1}_{C}D(t,T))=\mathbb{E}%
_{Q}(\mathbf{1}_{C}D(T,T))
\]
for every $C\in\mathcal{G}_{t}$. We can conclude that%
\[
S(t)=\mathbb{E}_{Q}(S(T)|\mathcal{G}_{t})\quad\text{and}\quad
D(t,T)=\mathbb{E}_{Q}(D(T,T)|\mathcal{G}_{t}),
\]
as required.
\end{proof}

\section{Example}

One simple example is the hazard process model with constant hazard rate
$\lambda>0$. In this case the default time~$\tau$ is exponentially distributed
under~$Q$ with parameter~$\lambda$, the survival process is given by
$G(t)=e^{-\lambda t}$, and the defaultable bond price is $D(t,T)=\mathbf{1}%
_{\{t<\tau\}}e^{-(r+\lambda)\left(  T-t\right)  }$.

We modify this simple example by stipulating two constants $\lambda
_{+},\lambda_{-}>0$, and taking the hazard rate to be the process $\lambda(t)$
equal to $\lambda_{+}$ whenever $W(t)\geq0$ and $\lambda_{-}$ otherwise. This
gives%
\[
G(t)=e^{-\int_{0}^{t}\lambda(s)ds}=e^{-\lambda_{+}\gamma_{+}(t)-\lambda
_{-}\gamma_{-}(t)}=e^{-\left(  \lambda_{+}-\lambda_{-}\right)  \gamma
_{+}(t)-\lambda_{-}t},
\]
where%
\[
\gamma_{+}(t):=\int_{0}^{t}\mathbf{1}_{[0,\infty)}(W(s))ds,\quad\gamma
_{-}(t):=\int_{0}^{t}\mathbf{1}_{(-\infty,0]}(W(s))ds
\]
are the sojourn times of$~W(t)$ above and below~$0$, which satisfy $\gamma
_{+}(t)+\gamma_{-}(t)=t$. Since $W(t)$ is the Brownian motion driving the
stock price process $S(t)$, this means that the hazard rate depends on whether
$S(t)$ is above or below a certain level.

The survival process $G(t)$ gives rise to a martingale measure~$Q$ as in the
construction in Section~\ref{Sect:hf64bs9a}. The probability distribution of
the default time~$\tau$ under~$Q$ can be found by computing the expectation%
\[
Q(t<\tau)=Q_{BS}(G(t))=e^{-\lambda_{-}t}\mathbb{E}_{Q_{BS}}(e^{-\left(
\lambda_{+}-\lambda_{-}\right)  \gamma_{+}(t)}).
\]
A formula for the Laplace transform of the probability distribution of the
sojourn time $\gamma_{+}(t)$ can be found, for example, in~\cite{BorSal2002},
part~II, formula~1.1.4.3. It gives%
\[
Q(t<\tau)=e^{-\frac{\left(  \lambda_{+}+\lambda_{-}\right)  t}{2}}I_{0}\left(
\frac{\left(  \lambda_{+}-\lambda_{-}\right)  t}{2}\right)  ,
\]
where
\[
I_{0}(x)=\frac{1}{\pi}\int_{0}^{\pi}e^{x\cos\theta}d\theta
\]
is the modified Bessel function of the first kind. Hence the price of the
defaultable bond at time~$0$ is%
\[
D(0,T)=e^{-rT}Q(T<\tau)=e^{-rT}e^{-\frac{\left(  \lambda_{+}+\lambda
_{-}\right)  T}{2}}I_{0}\left(  \frac{\left(  \lambda_{+}-\lambda_{-}\right)
T}{2}\right)  .
\]%
\begin{figure}
[t]
\begin{center}
\includegraphics[
height=2.6351in,
width=4.8205in
]%
{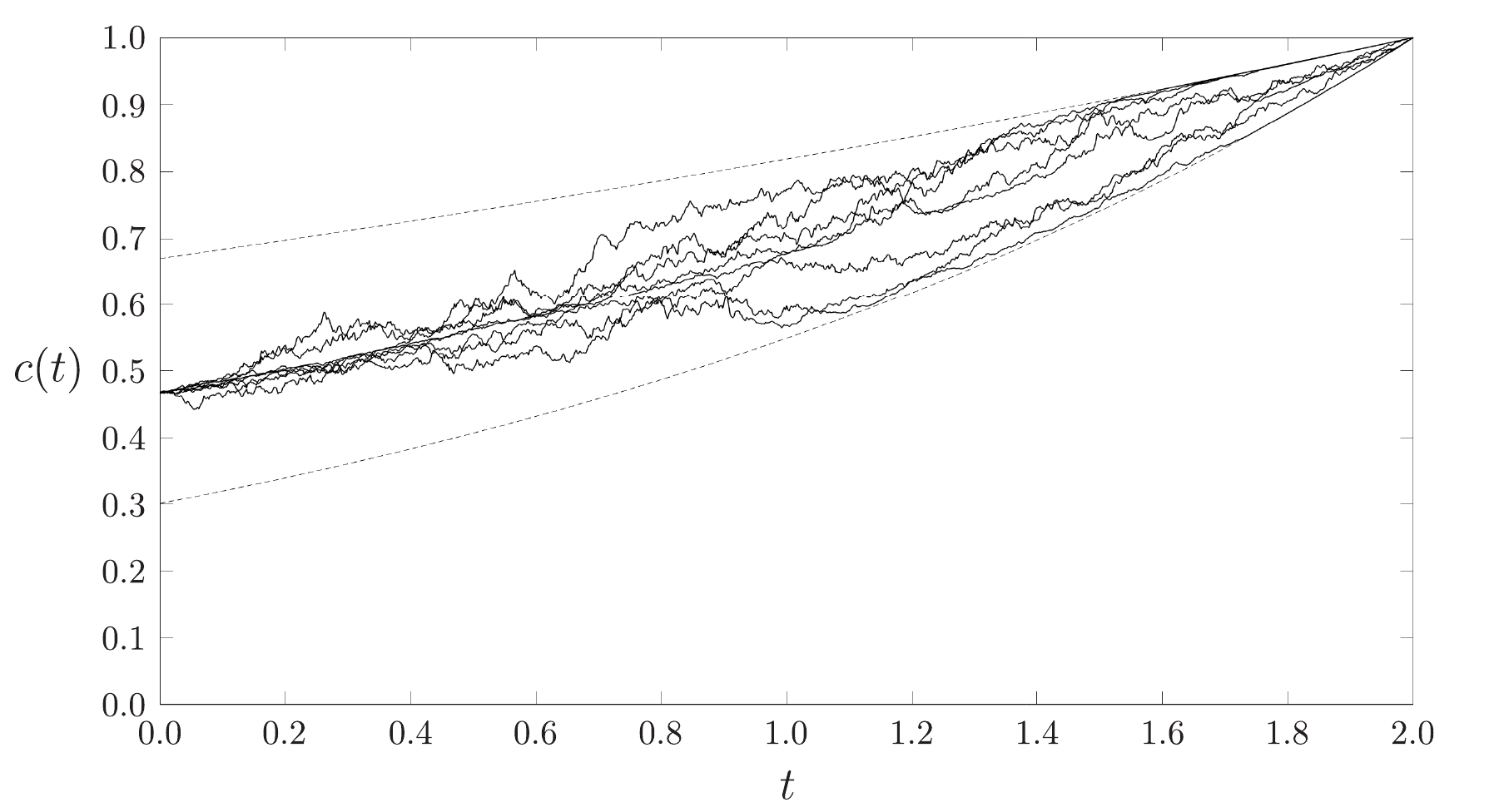}%
\caption{Sample paths of pre-default value $c(t)$ of $D(t,T)$}%
\label{Fig:gd6sgfsf}%
\end{center}
\end{figure}

It is also interesting to compute the pre-default value~$c(t)$, hence the
defaultable bond price $D(t,T)=\mathbf{1}_{\{t<\tau\}}c(t)$ for any
$t\in\lbrack0,T]$. We have%
\begin{align*}
c(t)  &  =e^{-r\left(  T-t\right)  }G(t)^{-1}\mathbb{E}_{Q_{BS}}%
(G(T)|\mathcal{F}_{t})\\
&  =e^{-\left(  r+\lambda_{-}\right)  \left(  T-t\right)  }\mathbb{E}_{Q_{BS}%
}(e^{-\left(  \lambda_{+}-\lambda_{-}\right)  \left(  \gamma_{+}(T)-\gamma
_{+}(t)\right)  }|\mathcal{F}_{t}).
\end{align*}
Observe that%
\[
\gamma_{+}(T)-\gamma_{+}(t)=\int_{t}^{T}\mathbf{1}_{\{W(s)\geq0\}}ds=\int
_{t}^{T}\mathbf{1}_{\{W(s)-W(t)\geq-W(t)\}}ds
\]
can be regarded as the sojourn time above $-W(t)$ of the Brownian motion
$W(s)-W(t)$ starting at time~$t$. Hence formula~1.1.4.3 in part~II
of~\cite{BorSal2002} for the Laplace transform of the probability distribution
of the sojourn time above a given level makes it possible to compute the above
conditional expectation, and gives%
\[
c(t)=e^{-\left(  r+\lambda_{+}\right)  \left(  T-t\right)  }\left[
\operatorname{erf}\left(  \frac{\left\vert W(t)\right\vert }{\sqrt{2\left(
T-t\right)  }}\right)  +\frac{1}{\pi}\int_{t}^{T}\frac{e^{\left(  \lambda
_{+}-\lambda_{-}\right)  \left(  T-s\right)  }e^{-\frac{W(t)^{2}}{2\left(
s-t\right)  }}}{\sqrt{(T-s)\left(  s-t\right)  }}ds\right]
\]
if $W(t)\geq0$, and%
\[
c(t)=e^{-\left(  r+\lambda_{-}\right)  \left(  T-t\right)  }\left[
\operatorname{erf}\left(  \frac{\left\vert W(t)\right\vert }{\sqrt{2(T-t)}%
}\right)  +\frac{1}{\pi}\int_{t}^{T}\frac{e^{\left(  \lambda_{-}-\lambda
_{+}\right)  \left(  T-s\right)  }e^{-\frac{W(t)^{2}}{2\left(  s-t\right)  }}%
}{\sqrt{(T-s)\left(  s-t\right)  }}ds\right]
\]
if $W(t)\leq0$, where%
\[
\operatorname{erf}(x)=\frac{1}{\sqrt{\pi}}\int_{-x}^{x}e^{-t^{2}}dt
\]
is the error function.

Several sample paths of the pre-default value $c(t)$ of the defaultable bond
$D(t,T)$ are shown in Figure~\ref{Fig:gd6sgfsf} for $T=2$, $r=0.1$, and
$\lambda_{+}=0.5$, $\lambda_{-}=0.1$. The broken lines marking the envelope of
the set of sample paths are the graphs of $e^{-\left(  r+\lambda_{+}\right)
\left(  T-t\right)  }$ and $e^{-\left(  r+\lambda_{-}\right)  \left(
T-t\right)  }$. In particular, we arrive at the price $D(0,T)=c(0)=0.4675$ for
the defaultable bond at time~$0$.

\section{Appendix: The form of $D(t,T)$\label{Sect:str5eydvshs}}

Here we show that expression (\ref{ahewtfd54bd7}) for the defaultable bond can
be obtained from some weaker assumptions about $D(t,T)$ and the lack of
arbitrage in a class of simple strategies in the $BD$ section of the market
only. This is similar to the hazard function model considered
in~\cite{CapZas2014}.

\begin{definition}
\label{Def:8fjd64hsan}\upshape By a $BD$\emph{-simple self-financing strategy}
we understand an $\mathbb{R}^{2}$-valued $(\mathcal{G}_{t})_{t\in\lbrack0,T]}%
$-adapted process $\psi=\left(  \psi^{B},\psi^{D}\right)  $ representing
positions in $B$ and~$D$ such that there are sequences of times $0=s_{0}%
<s_{1}<\cdots<s_{N}=T$ and random variables $x_{1},\ldots,x_{N}$ and
$y_{1},\ldots,y_{N}$ with the following properties:

\begin{enumerate}
\item $x_{n}$ and $y_{n}$ are $\mathcal{G}_{s_{n-1}}$-measurable and%
\[
\psi^{B}(t)=x_{n},\quad\psi^{D}(t)=y_{n}%
\]
for each $n=1,\ldots,N$ and $t\in(s_{n-1},s_{n}]$;

\item The \emph{value process}
\[
V_{\psi}(t):=\psi^{B}(t)B(t,T)+\psi^{D}(t)D(t,T)
\]
satisfies the following \emph{self-financing condition} for each
$n=0,\ldots,N-1$:%
\[
V_{\psi}(s_{n})=\lim_{t\searrow s_{n}}V_{\psi}(t).
\]

\end{enumerate}
\end{definition}

\begin{definition}
\label{Def:nv757ak0}\upshape We say that the \emph{no-}$BD$%
-\emph{simple-arbitrage} (NBDSA) \emph{principle} holds if there is no
$BD$-simple self-financing strategy~$\psi=\left(  \psi^{B},\psi^{D}\right)  $
such that $V_{\psi}(0)=0$, $V_{\psi}(T)\geq0$, and $V_{\psi}(T)>0$ with
positive probability~$P$.
\end{definition}

\begin{proposition}
\label{Prop:fg46std7sg}Suppose that $D(t,T)$ is a $\left(  \mathcal{G}%
_{t}\right)  _{t\in\lbrack0,T]}$-adapted process satisfying $D(T,T)=\mathbf{1}%
_{\left\{  T<\tau\right\}  }$, with paths which are continuous on
$[0,\tau)\cap\lbrack0,T]$ and right-continuous elsewhere. If the
{\upshape NBDSA} principle holds, then there is an $\left(  \mathcal{F}%
_{t}\right)  _{t\in\lbrack0,T]}$-adapted process $c(t)$ defined for all
$t\in\lbrack0,T]$, with continuous paths, such that $c(t)\in(0,1)$ for all
$t\in\lbrack0,T)$, $c(T)=1$ and%
\begin{equation}
D(t,T)=c(t)\mathbf{1}_{\left\{  t<\tau\right\}  } \label{dhd645avs8}%
\end{equation}
for all $t\in\lbrack0,T]$.
\end{proposition}

\begin{proof}
Because we can switch to working with discounted values, it is enough to
consider the case when $r=0$, so that $B(t,T)=1$ for all $t\in\lbrack0,T]$.

For $t=T$, we have $D(T,T)=c(T)1_{\{T<\tau\}}$ with $c(T)=1$. For any
$t\in\lbrack0,T)$, since $D(t,T)$ is a $\mathcal{G}_{t}$-measurable random
variable, it follows by well-known properties of the enlarged filtration (for
example, Proposition~5.28 in~\cite{CapZas2016}) that there exists an
$\mathcal{F}_{t}$-measurable random variable~$c(t)$ such that
\begin{equation}
D(t,T)\mathbf{1}_{\left\{  t<\tau\right\}  }=c(t)\mathbf{1}_{\left\{
t<\tau\right\}  }. \label{fhf648dbsh}%
\end{equation}
Let
\[
A_{t}:=\left\{  D(t,T)\mathbf{1}_{\left\{  \tau\leq t\right\}  }>0\right\}
,\quad A_{t}^{\prime}:=\left\{  D(t,T)\mathbf{1}_{\left\{  \tau\leq t\right\}
}<0\right\}  .
\]
We consider the $BD$-simple self-financing strategy%
\[%
\begin{array}
[c]{ll}%
\psi^{B}(u):=\psi^{D}(u):=0 & \text{for }u\in\lbrack0,t],\\
& \\
\psi^{B}(u):=D(t,T)\left(  \mathbf{1}_{A_{t}}-\mathbf{1}_{A_{t}^{\prime}%
}\right)  , & \\
\psi^{D}(u):=-\mathbf{1}_{A_{t}}+\mathbf{1}_{A_{t}^{\prime}} & \text{for }%
u\in(t,T].
\end{array}
\]
To verify that this is indeed a $BD$-simple self-financing strategy, when
$t\in(0,T)$, we take $N:=2$, $s_{0}:=0$, $s_{1}:=t$, $s_{2}:=T$ and $x_{1}%
:=0$, $y_{1}:=0$, $x_{2}:=D(t,T)\left(  \mathbf{1}_{A_{t}}-\mathbf{1}%
_{A_{t}^{\prime}}\right)  $, $y_{2}:=-\mathbf{1}_{A_{t}}+\mathbf{1}%
_{A_{t}^{\prime}}$ in Definition~\ref{Def:8fjd64hsan} to obtain this strategy.
When $t=0$, we take $N:=1$, $s_{0}:=0$, $s_{1}:=T$ and $x_{1}:=D(0,T)\left(
\mathbf{1}_{A_{0}}-\mathbf{1}_{A_{0}^{\prime}}\right)  $, $y_{1}%
:=-\mathbf{1}_{A_{0}}+\mathbf{1}_{A_{0}^{\prime}}$.

The initial value of this strategy is $V_{\psi}(0)=0$, and its final value
$V_{\psi}(T)=D(t,T)\left(  \mathbf{1}_{A_{t}}-\mathbf{1}_{A_{t}^{\prime}%
}\right)  $ is strictly positive on $A_{t}\cup A_{t}^{\prime}$ and $0$
otherwise. This would violate the {\upshape NBDSA} principle unless
$P(A_{t}\cup A_{t}^{\prime})=0$. Hence we can conclude that $D(t,T)\mathbf{1}%
_{\left\{  \tau\leq t\right\}  }=0$.\ Together with~(\ref{fhf648dbsh}), this
gives (\ref{dhd645avs8}).

We claim that $c(t)\in(0,1)$ for each $t\in\lbrack0,T)$. We put
\[
B_{t}:=\left\{  c(t)\leq0\right\}  ,\quad B_{t}^{\prime}:=\left\{
c(t)\geq1\right\}
\]
and observe that%
\[%
\begin{array}
[c]{ll}%
\varphi^{B}(u):=\varphi^{D}(u):=0 & \text{for }u\in\lbrack0,t],\\
& \\
\varphi^{B}(u):=c(t)\left(  \mathbf{1}_{B_{t}^{\prime}}-\mathbf{1}_{B_{t}%
}\right)  \mathbf{1}_{\left\{  t<\tau\right\}  }, & \\
\varphi^{D}(u):=-\left(  \mathbf{1}_{B_{t}^{\prime}}-\mathbf{1}_{B_{t}%
}\right)  \mathbf{1}_{\left\{  t<\tau\right\}  } & \text{for }u\in(t,T]
\end{array}
\]
is a $BD$-simple self-financing strategy with initial value $V_{\varphi}(0)=0$
and final value%
\begin{align*}
V_{\varphi}(T)  &  =c(t)\left(  \mathbf{1}_{B_{t}^{\prime}}-\mathbf{1}_{B_{t}%
}\right)  \mathbf{1}_{\left\{  t<\tau\right\}  }-\left(  \mathbf{1}%
_{B_{t}^{\prime}}-\mathbf{1}_{B_{t}}\right)  \mathbf{1}_{\left\{
T<\tau\right\}  }\\
&  =c(t)\left(  \mathbf{1}_{B_{t}^{\prime}}-\mathbf{1}_{B_{t}}\right)
\mathbf{1}_{\left\{  t<\tau\leq T\right\}  }+\left(  c(t)-1\right)  \left(
\mathbf{1}_{B_{t}^{\prime}}-\mathbf{1}_{B_{t}}\right)  \mathbf{1}_{\left\{
T<\tau\right\}  }\geq0.
\end{align*}
Given that the {\upshape NBDSA} principle holds, we must have $P(B_{t}%
\cap\left\{  T<\tau\right\}  )=0$ because $\left(  c(t)-1\right)  \left(
\mathbf{1}_{B_{t}^{\prime}}-\mathbf{1}_{B_{t}}\right)  >0$ on~$B_{t}$.
Moreover, we must have $P(B_{t}^{\prime}\cap\left\{  t<\tau\leq T\right\}
)=0$ since $c(t)\left(  \mathbf{1}_{B_{t}^{\prime}}-\mathbf{1}_{B_{t}}\right)
>0$ on~$B_{t}^{\prime}$. By assumption~(\ref{d7465edgd}), since $B_{t}%
,B_{t}^{\prime}\in\mathcal{F}_{t}\subset\mathcal{F}_{T}$, we obtain
$P(B_{t})=P(B_{t}^{\prime})=0$, which proves the claim.

Finally, we construct a continuous modification~$\hat{c}$ of~$c$, with
$\hat{c}(t)\in(0,1)$ for all $t\in\lbrack0,T)$ and $\hat{c}(T)=1$, adapted to
the filtration $(\mathcal{F}_{t})_{t\in\lbrack0,T]}$ and such that
$D(t,T)=\hat{c}(t)\mathbf{1}_{\{\tau<t\}}$ for all $t\in\lbrack0,T]$. Let%
\[
\mathbb{Q}_{T}:=\{Tq:q\in\mathbb{Q}\cap\lbrack0,1]\},
\]
and let%
\[
C:=\left\{  c\text{ is uniformly continuous on }\mathbb{Q}_{T}\right\}  .
\]
Because $\mathbb{Q}_{T}$ is countable, it follows that $C\in\mathcal{F}_{T}$.
Moreover, on $\left\{  T<\tau\right\}  $, we have $D(t,T)=c(t)$ for each
$t\in\mathbb{Q}_{T}$, and $D(t,T)$ is continuous, hence uniformly continuous
as a function of $t\in\lbrack0,T]$, so $c$ is uniformly continuous
on~$\mathbb{Q}_{T}$. This shows that $\left\{  T<\tau\right\}  \subset C$. As
a result, $P\left(  \left(  \Omega\setminus C\right)  \cap\left\{
T<\tau\right\}  \right)  =0$, so $P(\Omega\setminus C)=0$ by
assumption~(\ref{d7465edgd}). This means that $c$ is uniformly continuous
on~$\mathbb{Q}_{T}$ (almost surely under~$P$). Because $\mathbb{Q}_{T}$ is
dense in $[0,T]$, there exists a continuous extension~$\hat{c}$ of~$c$ onto
$[0,T]$. For each $t\in\lbrack0,T]$ we have $t\in\mathbb{Q}_{T}$ or there is a
sequence $t_{n}\in\mathbb{Q}_{T}$ such that $t_{n}<t$ and $t_{n}\rightarrow t$
as $n\rightarrow\infty$, so that $c(t_{n})=\hat{c}(t_{n})\rightarrow\hat
{c}(t)$ as $n\rightarrow\infty$, hence $\hat{c}(t)$ is $\mathcal{F}_{t}%
$-measurable. This shows that $\hat{c}$ is an $\left(  \mathcal{F}_{t}\right)
_{t\in\lbrack0,T]}$-adapted process with continuous paths. We have
$D(t,T)=c(t)=\hat{c}(t)$ for each $t\in\mathbb{Q}_{T}\cap\lbrack0,\tau)$.
Because $D(t,T)$ is continuous on $[0,T]\cap\lbrack0,\tau)$ and $\hat{c}(t)$
is continuous, it follows that $D(t,T)=\hat{c}(t)$ for each $t\in
\lbrack0,T]\cap\lbrack0,\tau)$. We have shown that $D(t,T)=0$ on $\left\{
\tau\leq t\right\}  $, so this means that $D(t,T)=\hat{c}(t)\mathbf{1}%
_{\left\{  t<\tau\right\}  }$ for each $t\in\lbrack0,T]$. Moreover, $\hat
{c}(t)\in(0,1)$ for all $t\in\lbrack0,T)$ and $\hat{c}(T)=1$ since $c$ has the
same properties. This concludes the proof.
\end{proof}

\end{document}